\documentclass{amsart}

\theoremstyle{plain}
\newtheorem{theorem}{Theorem}
\newtheorem*{theorem*}{Theorem}
\newtheorem{proposition}[theorem]{Proposition}

\newtheorem*{corollary*}{Corollary}
\newtheorem{lemma}[theorem]{Lemma}

\theoremstyle{remark}
\newtheorem{remark}[theorem]{Remark}
\newtheorem*{remark*}{Remark}
\newtheorem{example}[theorem]{Example}

\theoremstyle{definition}
\newtheorem{definition}[theorem]{Definition}

\usepackage{color}

\usepackage[matrix,arrow,curve]{xy}

\usepackage{tikz}
\usetikzlibrary{calc}
\usetikzlibrary{arrows,matrix,positioning}
\usepackage{xfrac}
\usepackage{hhline}
\usepackage[ruled,vlined]{algorithm2e}

\newcommand{\field}[1][]{\mathbb{F}_{#1}}

\newcommand{\lclm}[1]{\left[#1\right]_\ell}

\newcommand{\matrixring}[2]{\mathcal{M}_{#1}(#2)}

\newcommand{\Fq}{\mathbb{F}_q}
\newcommand{\Fz}{\mathbb{F}_q(z)}

\begin{document}

\title[PGZ algorithm for differential convolutional codes]{Peterson-Gorenstein-Zierler algorithm for differential convolutional codes}


\author[G\'{o}mez-Torrecillas]{Jos\'{e} G\'{o}mez-Torrecillas}
\address{IEMath-GR and Department of Algebra, University of Granada}
\curraddr{}
\email{gomezj@ugr.es}
\thanks{Research supported by grants MTM2016-78364-P and PID2019-110525GB-I00 from Agencia Estatal de Investigaci\'on and FEDER, and by grant \textup{A-FQM-470-UGR18} from Consejer\'{\i}a de Econom\'{\i}a y Conocimiento de la Junta de Andaluc\'{\i}a and Programa Operativo FEDER 2014-2020. 
 }
 \author[Lobillo]{F. J. Lobillo}
\address{CITIC and Department of Algebra, University of Granada}
\curraddr{}
\email{jlobillo@ugr.es}
\author[Navarro]{Gabriel Navarro}
\address{CITIC and Department of Computer Science and Artificial Intelligence, University of Granada}
\curraddr{}
\email{gnavarro@ugr.es}

\author[S\'{a}nchez-Hern\'{a}ndez]{Jos\'e Patricio S\'anchez-Hern\'andez}
\address{Department of Algebra, University of Granada}
\curraddr{}
\email{jpsanchez@correo.ugr.es}
\thanks{The fourth author was supported by The National Council of Science and Technology (CONACYT) with a scholarship  for a Postdoctoral Stay in the University of Granada.}

\subjclass[2010]{Primary }


\date{}

\dedicatory{}

\begin{abstract}
Differential Convolutional Codes with designed Hamming distance are defined, and an algebraic decoding algorithm, inspired by Peterson-Gorenstein-Zierler's algorithm, is designed for them.
\end{abstract}

\maketitle

\section{Introduction}
Linear convolutional codes of length $n$ can be understood (see \cite{Forney:1970, Johannesson/Zigangirov:1999, Rosenthal/Smarandache:1999}) as vector subspaces of $\Fz^n$, where $\Fz$ is the rational function field in the variable $z$ (which represents the delay operator) with coefficients in a field $\Fq$ with $q$ elements, with $q$ a power of a prime $p$.  Based upon this algebraic model, an alternative approach to cyclic convolutional codes to that of \cite{Piret:1976, Roos:1979} was proposed in \cite{Gomez/alt:2016}. While the first formalism led to general methods of construction of convolutional codes with cyclic structures \cite{GluesingSchmale:2004, Estrada/alt:2008, Lopez/Szabo:2013, Gomez/alt:2017a}, the second mathematical framework has been proved to be well suited to the design of efficient algebraic decoding algorithms \cite{Gomez/alt:2017b, Gomez/alt:2018a}. In the latter, the cyclic code is modeled from a left ideal of a suitable factor ring of a skew polynomial ring $\Fz[x;\sigma]$, where $\sigma$ is a field automorphism of $\Fz$. This is a non commutative polynomial ring, where the multiplication of the variable $x$ and the coefficients $a \in \Fz$ are twisted according to the rule $xa = \sigma(a)x$. Here we explore another option, namely, the ring of differential operators $\Fz[x;\delta]$, built from a derivation $\delta$ of $\Fz$. Now, the multiplication obeys the rule $xa = ax + \delta(a)$. We obtain convolutional codes that become MDS with respect to the Hamming distance, which will be called Reed-Solomon differential convolutional codes. We see that the decoding algorithm of Peterson-Gorenstein-Zierler  type from \cite{Gomez/alt:2018a} can be adapted to these differential convolutional codes. The proof of the correctness of this decoding algorithm in the differential setting does not follow straightforwardly from the arguments from \cite{Gomez/alt:2018a}, so we develop in this paper the mathematical tools needed to check it. We also include several examples, whose computations are done with the help of the Computer Algebra software SageMath \cite{sage}.

\section{Convolutional codes and distributed storage.}\label{DS}

Let us briefly recall from  \cite{Forney:1970, Roos:1979, GluesingSchmale:2004, Gomez/alt:2018b, Gomez/alt:2019} the notion of convolutional code.  Let \(\Fq[z] \subseteq \Fq(z) \subseteq \Fq((z))\) denote the ring of polynomials, field of rational functions and field of Laurent series over the finite field \(\Fq\), respectively. A rate \(k/n\) convolutional code is a \(k\)-dimensional vector subspace of \(\Fq((z))^n\) generated by a matrix with entries in \(\Fq(z)\). The idea behind the use of polynomials and Laurent series comes from the identification \(\Fq[z]^n \cong \Fq^n[z]\) (resp. \(\Fq((z))^n \cong \Fq^n((z))\)), i.e. its elements can be viewed as vectors of polynomials (resp. vectors of infinite sequences) or as polynomials of vectors (resp. infinite sequences of vectors) in \(\Fq\). Vectors in \(\Fq\) are usually referred to as words and a sequence of words as a sentence. The encoders (generator matrices) transform information sequences into code sequences. Formally, an information sequence is 
\(
\sum_{i} u_i z^i, \ u_i \in \Fq^k,
\)
and, by multiplication by a matrix 
\(G(z) \in \matrixring{k\times n}{\Fq(z)}\)
a sequence in the code is obtained
\(
\sum_{i} v_i z^i, \ v_i \in \Fq^n.
\)
Rational functions are used because they can be realized as linear shift registers with feedback (see \cite[Figure 2.1]{Johannesson/Zigangirov:1999}). From a mathematical viewpoint, the existence of a basis whose vectors belong to \(\Fq(z)^n\) implies that a rate \(k/n\) convolutional code can be equivalently defined as \(k\)-dimensional subspace of \(\Fq(z)^n\). In each communication process, the transmitted information is finite, so it is encoded by polynomials better than Laurent series. It is convenient to analyze convolutional codes in terms of polynomials. As pointed out in \cite[Proposition 1]{Gomez/alt:2019}, the map \(\mathcal{D} \mapsto \mathcal{D} \cap \Fq[z]^n\) defines a bijection between the set of rate \(k/n\) convolutional codes and the set of \(\Fq[z]\)-submodules of \(\Fq[z]^n\) of rank \(k\) that are direct summands of \(\Fq[z]^n\). This is a module-theoretical refinement of \cite[Theorem 3]{Forney:1970}.

There is a scenario in which Hamming distance in convolutional codes is the right one: The distributed storage. Let us explain it. 

We have some information which we want to store in \(n\) nodes. Different nodes are not usually accessible with the same reliability, depending on the different conditions at each time. So we want to distribute our information taking in mind that we are going to retrieve it from a subset of the nodes, although we are going to request the information stored in each node. The storage process can be described as follows. 

Let \(\sum_{i} v_i z^i \in \field{}^n[z]\) be the encoded sequence, and \(v_i = (v_{i,1}, \dots, v_{i,n}) \in \field{}^n\) for each \(i\). Then the encoded sequence can be reorganized as \((f_1, \dots, f_n) \in \field{}[z]^n\) where \(f_j = \sum_{i} v_{i,j} z^i\). Each polynomial \(f_j\) is transmitted to a different node, where it is stored. When the information is needed, the owner request to each node its polynomial, which is transmitted back. This is the key point, each polynomial comes from a different channel, so the probability of errors is different and it can be expected that the Hamming weight  fits better than the free distance in this scenario. 

There are some previous solutions to this process. For instance, in the distributed storage in an open network each channel has its own correction capability, so each channel can be considered as an erasure channel, i.e. an error in the transmission means that the corresponding information is not received. Under that viewpoint, \cite{Balaji/alt:2018} has recently provided a complete overview of different techniques involved in solving the distributed storage problem. None of those techniques makes use of convolutional codes, mainly because the most used decoding algorithms for convolutional codes work with respect to the free distance, which is not suitable for distributed storage. 

The idea behind the use of non commutative structures when dealing with convolutional codes comes from \cite{Piret:1976}, where it is proven that a classical cyclic structure on words does not allow to build non block codes. However skew structures do, as finally clarified in \cite[Theorems 3.12 and 3.14]{Lopez/Szabo:2013}. With this idea in mind, in \cite{Gomez/alt:2016} a new perspective of cyclicity is provided when considering convolutional codes as vector subspaces of \(\Fq(z)^n\). This viewpoint is the one we are going to exploit in the forthcoming sections.

\section{Differential convolutional codes}

In this section, differential convolutional codes are to be defined. They are built from a derivation of the rational function field $\Fz$, where $\Fq$ is a finite field of characteristic $p$. We include some basic facts on these derivations that are needed subsequently. We next construct the codeword ambient algebra as a factor algebra of a differential operator ring in one variable. We thus recall some basic properties of this kind of rings.

Fix a non zero derivation $\delta : \Fz \to \Fz$, that is, an additive map subject to the condition 
\begin{equation}\label{Leibniz}
\delta(ab)=a\delta(b)+\delta(a)b, \qquad \text{ for all } a, b \in \Fz. 
\end{equation}
Since $a^{q-1} = 1$ for every nonzero $a \in \Fq$, we get from \eqref{Leibniz} that $\delta(a) = 0$ and, therefore, $\delta$ is $\Fq$--linear. A straightforward argument, based also on \eqref{Leibniz}, shows that, for every $f \in \Fz$,
\begin{equation}
\delta (f) = f'\delta(z),
\end{equation}
where $f'$ denotes the usual derivative of $f$ with respect to $z$. Hence, the derivation $\delta$ is determined by the choice of a non zero $\delta(z) \in \Fz$. 

Let $\mathbb{F}_q(z^p)$ denote the subfield of $\Fz$ generated by $\Fq$ and $z^p$. Clearly, $\mathbb{F}_q(z^p)$ is contained in the subfield of constants of $\delta$,  defined as 
\[
K = \{ f \in \Fz : \delta(f) = 0 \}.
\]
 Since the degree of $\Fz$ over $\mathbb{F}_q(z^p)$ is the prime $p$, we get that $K = \mathbb{F}_q(z^p)$. Obviously, $\delta : \Fz \to \Fz$ becomes a $K$--linear endomorphism, 
whose minimal polynomial is, by \cite[Lemma 1.5.1, Lemma 1.5.2]{Jacobson:1996},  of the form $x^p-\gamma x$ for a suitable $\gamma \in K$. Therefore, 
\begin{equation}\label{polmindelta}
\delta^p - \gamma \delta = 0,
\end{equation}
which implies that $\gamma = \frac{\delta^p(z)}{\delta(z)}$. Since $p$ equals the dimension of $\Fz$ as a $K$--vector space, we can choose a cyclic vector $\alpha \in \Fz$ for $\delta$, so that $\{\alpha, \delta(\alpha), \dots, \delta^{p-1}(\alpha) \}$ is a $K$--basis of $\Fz$.

Given a subset $\{ c_1, \dots, c_n \} \subseteq \Fz$, define, following \cite{Lam/Leroy:1988}, the \emph{Wronskian matrices}  
\[
W_{k}(c_1,\dots,c_n)=\left( \begin{array}{ccccc}
c_1 &c_2 & \cdots &c_n \\
\delta(c_1)& \delta(c_2) & \cdots &  \delta(c_n)  \\
\vdots &\vdots & \ddots& \vdots\\
\delta^{k-1}(c_1)& \delta^{k-1}(c_2) & \cdots &  \delta^{k-1}(c_n)
\end{array} \right) 
\]
for each $k \geq 1$. We record the following well-known result  (see e. g. \cite[Theorem 4.9]{Lam/Leroy:1988})  for future reference. 

\begin{lemma}\label{lemma 3}
A subset $\{c _1, \dots, c_n \} \subseteq \Fz$ is linearly independent over $K$ if and only if the Wronskian $W_{n}(c_1, \dots, c_n)$
is an invertible matrix.
\end{lemma}

Let $R = \Fz[x;\delta]$ be the differential operator ring built from the derivation $\delta$. Its elements are polynomials in $x$ with coefficients, written of the left, from $\Fz$, and whose multiplication is based on the rule $xf = fx + \delta (f)$ for every $f \in \Fz$. The ring $R$ is a non-commutative Euclidean domain, that is, there are left and right Euclidean division algorithms. In particular, every left ideal of $R$ is principal, that is, of the form $Rf = \{rf : r \in R \}$ for some $f \in R$. We say then that $f$ is a generator of the left ideal $Rf$. Given $f, g \in R$, we say that $f$ is a \emph{right divisor} of $g$ (or $g$ is a \emph{left multiple} of $f$), if $g \in Rf$. This situation will be denoted by $f |_r g$.

The \emph{greatest common right divisor} of $f, g \in R$, denoted by $(f,g)_r$ is defined as the monic generator of $Rf + Rg$. Similarly, the \emph{left least common multiple} $[f,g]_\ell$ of $f$ and $g$ is defined as the monic generator of $Rf \cap Rg$.

The center of $R$ is, by \cite[1.12.32]{Jacobson:1996}, $K[x^p - \gamma x]$, the subring of $R$ generated by $K$ and $x^p - \gamma x$. In particular, we get that $R(x^p - \gamma x)$ is an ideal of $R$, and we can consider the $K$--algebra 
\[
\mathcal{R} = \frac{R}{R(x^p - \gamma x)}. 
\]

Now, let us fix the natural $\Fz$--basis $\{ 1, x , \dots, x^{p-1} \}$ of $\mathcal{R}$, and the corresponding coordinate isomorphism of $\Fz$--vector spaces 
\[
\mathfrak{v} : \mathcal{R} \to \Fz^p.
\]
 We stress that we are identifying an element of $\mathcal{R}$ with its unique representative in $R$ of degree smaller than $p$. Recall that the convolutional codes over $\Fq$ of length $p$ are the $\Fz$--vector subspaces of $\Fz^p$.

\begin{definition}
A convolutional code $C \subseteq \Fz^p$ is said to be a \emph{differential convolutional code}  if $\mathfrak{v}^{-1}(C)$ is a left ideal of $\mathcal{R}$. 
\end{definition}

Since the isomorphism $\mathfrak{v}$ is an isometry for the natural Hamming  distances on $\mathcal{R}$ and $\Fz^p$, we will often consider a differential convolutional code just as a left ideal of $\mathcal{R}$.  We stress that Hamming weight of a vector of $\Fz^p$ is the number of nonzero coordinates (see Section \ref{DS} for the description of a situation where the Hamming distance is better suited than the free distance).

We define by recursion, following \cite{Lam/Leroy:1988},  $N_k(a)$, for $a \in \Fz$ as follows:

$$N_0(a)=1,$$
$$N_{n+1}(a)=N_n(a)a+\delta(N_n(a)).$$

We have, for $0 \neq a \in \Fz$,  the following identity \cite[Proposition 2.9(4)]{Lam/Leroy:1988}: 

\begin{equation}\label{N_k(a)}
\delta^k(a)=N_k(L(a))a,
\end{equation}
where
\[
L(a) = \delta(a)a^{-1}.
\]

Let  $g(x)=\sum_{i=0}^{n}g_ix^i\in \Fz[x;\delta]$ and $a \in \Fz$.  By \cite[Lemma 2.4]{Lam/Leroy:1988},  
\begin{equation}\label{root1}
g(x)=q(x)(x-a)+g[a],
\end{equation}
 where
\begin{equation}\label{root2}
g[a] =\sum_{i=0}^{n}g_iN_i(a) \in \Fz.
\end{equation}
The value $g[a]$ given in \eqref{root2} is called \emph{right evaluation of $g$ at $a$}. Observe that $x-a$ is a right divisor of  $g(x)$ if and only if $g[a] = 0$. We say then that $a$ is a \emph{right root} of $g(x)$.  A straightforward computation, using that $\delta^p = \gamma \delta$, shows that $L(c)$ is a right root of $x^p-\gamma x$ for every $0 \neq c \in \Fz$.

\begin{proposition}\label{DCCdim}
Let $\{ c_1, \dots, c_m \} \subseteq \Fz$ be linearly independent over $K$, with $m \leq p-1$, and set  
\[
g = [x-L(c_1), \dots, x-L(c_m)]_\ell.
\]
Then $deg(g) = m$, $g$ is a right divisor of $x^p- \gamma x$,  and $\mathfrak{v}(\mathcal{R}g)$ is the left kernel of the Wronskian matrix 
\[
W_{p}(c_1,\dots,c_m).
\]
In other words, $C =  \mathfrak{v}(\mathcal{R}g)$ is a DCC of dimension $p-m$ with parity-check matrix $W_{p}(c_1,\dots,c_m)$. 
\end{proposition}
\begin{proof}
Clearly, $f=\sum_{k=0}^{p-1}f_kx^k\in \mathcal{R}g$ if and only if $x-L(c_j)|_r f$ for all $j=1,\dots,m$. This is equivalent, by \eqref{root1} and \eqref{root2}, to the condition 
$(f_0,\dots, f_{p-1})N=0$, where
 $$N= \left( \begin{array}{ccccc}
N_0(L(c_1)) & N_0(L(c_2)) & \cdots &N_{0}(L(c_m))  \\
N_1(L(c_1))&N_1(L(c_2))& \cdots & N_{1}(L(c_m))  \\
\vdots &\vdots & \ddots& \vdots\\
N_{p-1}(L(c_1))& N_{p-1}(L(c_2)) & \cdots &  N_{p-1}(L(c_m))
\end{array} \right). $$
By \eqref{N_k(a)},
\begin{equation}\label{NW}
N \cdot diag(c_1,\dots,c_m) = W_{p}(c_1, \dots, c_m).
\end{equation}
Since $diag(c_1, \cdots, c_m)$ is invertible, we get that the left kernel of the wronskian matrix $W_{p}(c_1, \dots, c_m)$ is $C = \mathfrak{v}(\mathcal{R}g)$.  Now, if the degree of $f$ is smaller than $m$, we get that 
$\mathfrak{v}(f) W_{m}(c_1, \dots, c_m) = 0$. By Lemma \ref{lemma 3}, $W_{m}(c_1, \dots, c_m)$ is invertible, which implies that $\mathfrak{v}(f) = 0$. We thus deduce that every nonzero $f \in \mathcal{R}g$ is of degree at least $m$. In particular, $deg (g) = m$. Since $dim_{\Fz} (\mathcal{R}/\mathcal{R}g) = deg(g)$, we get that the dimension of $C$ is $p-m$. 

\end{proof}

\begin{remark}
Differential convolutional codes described in Proposition \ref{DCCdim} are instances of $(\sigma,\delta)-W$--codes in the sense of \cite{Boulagouaz/Leroy:2013}. These codes are recognized as left kernels of suitable generalized Vandermonde matrices in \cite[Proposition 4]{Boulagouaz/Leroy:2013}. 
\end{remark}

Next, we show that, if the set $\{c_1, \dots, c_m\}$ is carefully chosen, then we obtain MDS codes with respect to  the Hamming distance. 

\begin{theorem} \label{Teo 1} Let $\alpha \in \Fz$ be a cyclic vector for $\delta$. For $1 \leq d \leq p$ and $0 \leq r \leq p-d$, let $C$ be the differential convolutional code generated, as a left ideal of $\mathcal{R}$, by
\[
g = [x-L(\delta^{r}(\alpha)),\dots,x-L(\delta^{r+d-2}(\alpha))]_\ell.
\]
Then the  dimension of $C$ is $p- d+ 1$ and its minimum Hamming distance is $d$. 
\end{theorem}
\begin{proof}
The statement about the dimension of $C$ follows from Proposition \ref{DCCdim}, which also says that a parity-check matrix of $C$ is 
\[
W_{p}(\delta^r(\alpha), \dots, \delta^{r+d-2}(\alpha)). 
\]
For any submatrix of order $d-1$ 
$$M = \left( \begin{array}{ccccc}
\delta^{k_1 + r}(\alpha) & \delta^{k_1+r+1}(\alpha)&  \cdots &\delta^{k_1 + r+d-2}(\alpha)  \\
\delta^{k_2+ r}(\alpha)& \delta^{k_2+ r+1}(\alpha)& \cdots & \delta^{k_2+ r+d-2}(\alpha)  \\
\vdots &\vdots & \ddots& \vdots\\
\delta^{k_{d-1}+ r}(\alpha)& \delta^{k_{d-1}+ r+1}(\alpha) & \cdots & \delta^{k_{d-1} + r+d-2}(\alpha)
\end{array} \right), $$
 where $\{k_1,\dots, k_{d-1}\}\subseteq\{0,\dots,p-1\}$, we see that
 \[
 M = W_{d-1}(\delta^{k_1 + r}{(\alpha)},\dots,\delta^{k_{d-1}+ r}{(\alpha}))^{tr},
  \]
 which is, by Lemma \ref{lemma 3}, invertible. Hence, the Hamming distance of $C$ is $d$. 
 \end{proof}

\begin{definition}
We call the code defined in Theorem \ref{Teo 1} a  \emph{Reed Solomon}  (RS)  differential convolutional code  of designed Hamming  distance $d$. 
\end{definition}

\section{A Peterson-Gorenstein-Zierler decoding algorithm}
In this section we design a decoding algorithm for RS differential convolutional codes inspired by the classical Peterson-Gorenstein-Zierler decoding algorithm.

Let $C$ be an RS differential convolutional code of designed distance $d$ generated, as left ideal of $\mathcal{R}$, by 
\[
g= [x-L(\alpha),x-L(\delta(\alpha)),\dots,x-L(\delta^{d-2}(\alpha))]_\ell,
\]
where $\alpha \in \Fz$ is a cyclic vector for $\delta$. By Theorem \ref{Teo 1}, the Hamming distance of $C$ is  $d$, and  $\tau=\lfloor (d-1)/2 \rfloor$ is then the error correction capacity of $C$.  

Let $c\in C$ be a codeword that is transmitted through a noisy channel, and let 
\[
y= \sum_{j=0}^{p-1}y_jx^j \in \mathcal{R}
\]
be the received polynomial. We may decompose $y = c + e$, where
\[
e=e_1x^{k_1}+\cdots+e_vx^{k_v} \in \mathcal{R}
\]
 is the error polynomial. Assume that $v\leq \tau$.

For each $0 \leq i \leq p-1$, the $i-$th syndrome $s_i$ of the received polynomial is computed as 
\begin{equation}\label{synd}
s_i = y[L(\delta^i(\alpha))] =  \sum_{j=0}^{p-1}y_jN_j(L(\delta^i(\alpha))); 
\end{equation}
the remainder of the left division of $y$ by $x-L(\delta^i(\alpha))$, see \eqref{root1} and \eqref{root2}.

\begin{proposition} \label{error values}The error values vector $(e_1,\dots,e_v)$ satisfies, for  $ 0 \leq i \leq 2\tau-1$, 
\begin{equation}\label{2tau-1}
\delta^i(\alpha)s_i=\sum_{j=1}^{v}e_j\delta^{i+k_j}(\alpha)
\end{equation}
and, indeed, is the unique solution of the linear system
\begin{equation} \label{delta i s_i}
\delta^i(\alpha)s_i=\sum_{j=1}^{v}e_j\delta^{i+k_j}(\alpha), \qquad i = 0, \dots, v-1.
\end{equation}
\end{proposition}

\begin{proof}
Whenever $0\leq i \leq  2\tau-1$, we have, by \eqref{root1} and \eqref{root2}, that the right evaluations of $c$ at $L(\delta^i(\alpha))$ are zero and hence, 
\begin{equation*}
\begin{array}{rclr}s_i&=& \sum_{j=0}^{p-1}y_jN_j(L(\delta^i(\alpha))) & \\&
= &\sum_{j=1}^{v}e_jN_{k_j}(L(\delta^i(\alpha))) & 
\\
&=& \sum_{j=1}^{v}e_j\frac{\delta^{k_j+i}(\alpha)}{\delta^{i}(\alpha)} & \text{by  (\ref{N_k(a)}),} \\
&=& \frac{1}{\delta^{i}(\alpha)}\sum_{j=1}^{v}e_j\delta^{i+k_j}(\alpha). & 
\end{array}
\end{equation*}
Thus, $(e_1, \dots, e_v)$ is solution to \eqref{2tau-1}. Now, the coefficient matrix of \eqref{delta i s_i} is 
\[
\left( \begin{array}{ccccc}
\delta^{k_1}(\alpha) &\delta^{k_1+1}(\alpha)  && \cdots &\delta^{k_1+v-1}(\alpha)  \\
\delta^{k_2}(\alpha) & \delta^{k_2+1}(\alpha) && \cdots &  \delta^{k_2+v-1}(\alpha) \\
\vdots &\vdots && \ddots& \vdots\\
\delta^{k_v}(\alpha) & \delta^{k_v+1}(\alpha) && \cdots &  \delta^{k_v+v-1}(\alpha)
\end{array} \right)  = W_{v}(\delta_{k_1}(\alpha), \dots, \delta_{k_v}(\alpha))^{tr}
\]
which is invertible by Lemma \ref{lemma 3}. 
\end{proof}

So, in order to compute the error values vector, we only need to know the error positions $\{ k_1, \dots, k_v \}$. 

\begin{proposition}\label{locator}
A position $t \in \{k_1, \dots, k_v \}$ if, and only if, $L(\delta^t(\alpha))$ is a right root of the \emph{error locator polynomial}
$$\lambda=[x- L(\delta^{k_1}(\alpha)),\dots,x-L(\delta^{k_v}(\alpha))]_\ell.$$
\end{proposition}
\begin{proof}
Proposition \ref{DCCdim} says that $deg (\lambda) = v$. Moreover, if there is some $t \not \in \{k_1, \dots, k_v\}$ such that $x-L(\delta^t(\alpha))|_r\lambda$, then, by Proposition \ref{DCCdim}, $deg(\lambda)=v+1$, a contradiction. 
\end{proof}

In view of propositions \ref{error values} and \ref{locator}, the error polynomial $e$ is known as soon as we compute the locator polynomial $\lambda$. Next, we will design an algorithm to do this calculation. 

For every pair $(i,k)$ of non-negative integers,  set  
\begin{equation}
S_{i,k}=\sum_{j=1}^v\delta^k(e_j)\delta^{i+k_j}(\alpha).
\end{equation}

\begin{lemma}\label{calcular}
For every pair $(i,k)$ of non-negative integers, we have
\begin{equation}\label{Sik}
S_{i,k+1} = \delta(S_{i,k})-S_{i+1,k}
\end{equation}
Moreover, for $0 \leq k \leq \tau-1$ and $i+k \leq 2\tau -1$, the values $S_{i,k}$ can be computed from the received polynomial $y$.
\end{lemma}
\begin{proof}
For every $(i,k)$, 
\[
\begin{array}{rcll}S_{i,k+1}&=&\sum_{j=1}^v\delta^{k+1}(e_j)\delta^{i+k_j}(\alpha)&\\& =& \delta(\sum_{j=1}^v\delta^{k}(e_j)\delta^{i+k_j}(\alpha))   -\sum_{j=1}^v\delta^{k}(e_j)\delta^{i+1+k_j}(\alpha)& \\& =& \delta(S_{i,k})-S_{i+1,k}.& 
\end{array}
\]
If $k=0$ and $0 \leq i \leq  2\tau-1$, then, by (\ref{2tau-1}), 
\[
S_{i,0}=\sum_{j=1}^ve_j\delta^{i+k_j}(\alpha)=s_i\delta^i(\alpha).
\] 
If $k + 1 + i \leq 2\tau - 1$, $S_{i,k+1}$ is computed from \eqref{Sik}. 
\end{proof}

Recall that $\mathfrak{v} : \mathcal{R} \to \Fz^p$ denotes the coordinate isomorphism associated to the monomial basis $\{1, x, \dots, x^{p-1}\}$ of $\mathcal{R}$. From Proposition \ref{DCCdim} we know that $f=\sum_{k=0}^{p-1}f_kx^k\in \mathcal{R}\lambda$ if and only if $\mathfrak{v}(f)$ belongs to the left kernel of the matrix 

$$\Sigma=\left( \begin{array}{ccccc}
\delta^{k_1}(\alpha) &\delta^{k_2}(\alpha)  && \cdots &\delta^{k_v}(\alpha)  \\

\delta^{k_1+1}(\alpha) & \delta^{k_2+1}(\alpha) && \cdots &  \delta^{k_v+1}(\alpha) \\

\vdots &\vdots && \ddots& \vdots\\

\delta^{k_1+p-1}(\alpha) & \delta^{k_2+p-1}(\alpha) && \cdots &  \delta^{k_v+p-1}(\alpha)
\end{array} \right)
$$
\begin{proposition} \label{Prop 3} Define, for every $1 \leq r$, the matrix
\[ E^r=\left( \begin{array}{ccccc}
e_1 &\delta(e_1)  && \cdots &\delta^{r-1}(e_1)  \\
e_2 & \delta(e_2) && \cdots &  \delta^{r-1}(e_2) \\
\vdots &\vdots && \ddots& \vdots\\
e_v & \delta(e_v ) && \cdots &  \delta^{r-1}(e_v )
\end{array} \right)_{v\times r},
\]
and set 
\[
\mu = \max \{ r \leq v :  E^r \text{ is full rank}\}.
\]

If $V \subseteq \Fz^p$ is the left kernel of the matrix $\Sigma E^{\mu}$, then  $\mathfrak{v}^{-1}(V)=\mathcal{R}\rho$ for some $\rho\in \mathcal{R}$ of degree $\mu$. Moreover, $\rho$ is a right divisor of $\lambda$.
\end{proposition}
\begin{proof}
We need to prove that $\mathfrak{v}^{-1}(V)$ is a left ideal of $\mathcal{R}$ which, in coordinates, amounts to check that if 
\[ 
(a_0,\dots,a_{p-2},a_{p-1})\in V 
\] then 
\[ 
(\delta(a_0),a_0+\delta(a_1)+ \gamma a_{p-1},\dots,a_{p-2}+\delta(a_{p-1}))\in V.
\] 
This is due to the facts that $xa_ix^i=a_ix^{i+1}+\delta(a_i)x^i$ for all $0\leq i \leq p-1$, and $x^p = \gamma x$ in $\mathcal{R}$. 

Suppose that $(a_0,\dots,a_{p-2},a_{p-1})\Sigma E^{\mu}=0$.  The maximality of $\mu$ ensures that the last column of $E^{\mu+1}$ is a linear combination of the previous $\mu$ columns. Hence $(a_0,\dots,a_{p-2},a_{p-1})\Sigma E^{\mu+1}=0$.

Observe that \[ 
\Sigma E^{\mu +1 }=  
\left( \begin{array}{ccccc}
S_{0,0} & S_{0,1} && \cdots & S_{0,\mu} \\
S_{1,0} & S_{1,1} && \cdots & S_{1,\mu}  \\
\vdots &\vdots && \ddots& \vdots\\
S_{p-1,0} & S_{p-1,1} && \cdots &  S_{p-1,\mu}
\end{array} \right)
 \]
Therefore, 
\begin{equation} \label{sumas}
\sum_{i=0}^{p-1}a_iS_{i,k}=0, \qquad \text{ for all } 0 \leq k \leq \mu.
\end{equation}

Let $0\leq k \leq \mu-1$ and set $a_{-1}=0$. Then
$$\begin{array}{rcll}
\sum_{i=0}^{p-1}(a_{i-1}+\delta(a_i))S_{i,k} &= & \sum_{i=0}^{p-1}a_{i-1}S_{i,k}+\delta(a_i)S_{i,k}  \\
&= & \sum_{i=0}^{p-1}a_{i-1}S_{i,k}+(\delta(a_iS_{i,k})-a_i\delta(S_{i,k}))\\
&= & \sum_{i=0}^{p-1}a_{i-1}S_{i,k}+\delta( \sum_{i=0}^{p-1}a_iS_{i,k})- \sum_{i=0}^{p-1}a_i\delta(S_{i,k})\\
&= & \sum_{i=0}^{p-1}a_{i-1}S_{i,k}- \sum_{i=0}^{p-1}a_i\delta(S_{i,k})& \\
&= & \sum_{i=0}^{p-1}a_{i-1}S_{i,k}- \sum_{i=0}^{p-1}a_i[S_{i,k+1}+S_{i+1,k}] \\
&= & \sum_{i=0}^{p-1}a_{i-1}S_{i,k}- \sum_{i=0}^{p-1}a_iS_{i,k+1}-\sum_{i=0}^{p-1}a_{i}S_{i+1,k} \\
&= & \sum_{i=0}^{p-1}a_{i-1}S_{i,k}-\sum_{i=0}^{p-1}a_{i}S_{i+1,k} \\
&= & \sum_{i=0}^{p-2}a_{i}S_{i+1,k}-\sum_{i=0}^{p-1}a_{i}S_{i+1,k} \\
&= & -a_{p-1}S_{p,k}  \\
\end{array}$$

The fourth and seventh steps hold by \eqref{sumas}, while the fifth does by  \eqref{Sik}, and the last equality uses that $a_{-1}=0$.
Since $\delta^p-\gamma \delta=0$, we have that  
\[
S_{p,k}=\sum_{j=0}^{v}\delta^k(e_j)\delta^{p+k_j}(\alpha)=\sum_{j=0}^{v}\delta^k(e_j)\gamma \delta^{k_j+1}(\alpha) =\gamma\sum_{j=0}^{v}\delta^k(e_j)\delta^{k_j+1}(\alpha)=\gamma S_{1,k}.
\]
Then $\sum_{i=0}^{p-1}(a_{i-1}+\delta(a_i))S_{i,k}=-\gamma a_{p-1}S_{1,k}$. Hence, 
\[
(\delta(a_0),a_0+\delta(a_1)+\gamma a_{p-1},\dots,a_{p-2}+\delta(a_{p-1}))\Sigma E^{\mu}=0,
\]
as required.

Every left ideal of $\mathcal{R}$ is principal, so $\mathfrak{v}^{-1}(V)$ is generated by a polynomial $\rho \in \mathcal{R}$. Since $\mathfrak{v}(\mathcal{R}\lambda)$ is the kernel of the matrix $\Sigma$, it follows that $\mathcal{R}\lambda\subseteq \mathcal{R}\rho$, hence $\rho$ right divides $\lambda$. 

Finally, the dimension of $\mathcal{R}\rho$ is, on the one hand, $p - \deg \rho$, and on the other hand, since $\Sigma E^\mu$ is a full rank matrix, $p - \mu$. Hence, $\deg \rho = \mu$. 
\end{proof}

In order to compute $\rho$ we should be, in principle, able to compute $\mu$ and the matrix $\Sigma E^{\mu}$ from the received polynomial. However, even if we knew $\mu$, we could only compute some of the coefficients of $\Sigma E^{\mu}$, see Lemma \ref{calcular}.

For each $r\leq \tau$, set $S^r=\Sigma^{\tau}E^r$, where 

$$\Sigma^{\tau}=\left( \begin{array}{ccccc}
\delta^{k_1}(\alpha) &\delta^{k_2}(\alpha)  && \cdots &\delta^{k_v}(\alpha)  \\
\delta^{k_1+1}(\alpha) & \delta^{k_2+1}(\alpha) && \cdots &  \delta^{k_v+1}(\alpha) \\
\vdots &\vdots && \ddots& \vdots\\
\delta^{k_1+\tau}(\alpha) & \delta^{k_2+\tau}(\alpha) && \cdots &  \delta^{k_v+\tau}(\alpha)
\end{array} \right)_{(\tau+1)\times v}. $$

Observe that 
 \[
 S^r=\left( \begin{array}{ccccc}
S_{0,0} & S_{0,1} && \cdots & S_{0,r-1} \\
S_{1,0} & S_{1,1} && \cdots & S_{1,r-1}  \\
\vdots &\vdots && \ddots& \vdots\\
S_{\tau,0} & S_{\tau,1} && \cdots &  S_{\tau,r-1}
\end{array} \right)_{(\tau+1)\times r}
 \]
 for all $1 \leq r \leq \tau$. Indeed, $S^r$ consists of the first $r$ columns of $S^\tau$ and, by Lemma \ref{calcular}, this matrix can be computed from the received polynomial $y$.

\begin{lemma}\label{lemma 4} 
For each $r\leq \tau$, $rk\, S^r=rk\, \Sigma E^r=rk\, E^r$. Consequently, 
\begin{equation}\label{mudos}
\mu = \max \{ r : S^r \text{ has full rank}\}.
\end{equation}
\end{lemma}

\begin{proof}
By Lemma \ref{lemma 3}, $rk\, \Sigma=rk\, \Sigma^{\tau} =v$. By Sylvester's rank inequality, 
$$
\min\{rk\, \Sigma, rk\, E^r\}\geq rk\, \Sigma E^r \geq rk\, \Sigma + rk\,  E^r - v=rk\,  E^r. 
$$
Then $rk\, \Sigma E^r=rk\, E^r$. Analogously,  $rk\, S^r =rk\, E^r$. 
Finally, since $\mu \leq v \leq \tau$, we get \eqref{mudos}.
\end{proof}

Thus, $\mu$ can be computed from $S^\tau$ by virtue of Lemma \ref{lemma 3}. Our next aim is to show how to use this matrix for the computation of $\rho$. 

\begin{lemma} \label{lemma 5} For each $\mu \leq r\leq \tau$, $\mu=rk\, E^r=rk\, \Sigma E^r=rk\, S^r$. 
\end{lemma}
\begin{proof} By Lemma \ref{lemma 4}, $rk\, S^r=rk\, \Sigma E^r = rk\, E^r$, for all $r \leq \tau$ and $rk\, E^{\mu} = \mu$. Assume that $\mu < r$. By maximality of $\mu$, the $(\mu+1)$th column of $E^r$ is a linear combination of the $\mu$ preceding columns. So, there exist $a_0,\dots,a_{\mu-1} \in \Fz$ such that, for each $1 \leq k \leq v$, 

\begin{equation}\label{eq mu}\delta^{\mu}(e_k)=\sum_{i=0}^{\mu-1}a_i\delta^{i}(e_k).
\end{equation}

Applying $\delta$, for each $1\leq k\leq v$, we have: 
$$\begin{array}{rcll}\delta^{\mu+1}(e_k)&=&\sum_{i=0}^{\mu-1}(\delta(a_i)\delta^{i}(e_k)+a_i\delta^{i+1}(e_k)) & \\
&=&\sum_{i=0}^{\mu-1}\delta(a_i)\delta^{i}(e_k)+\sum_{i=1}^{\mu}a_{i-1}\delta^{i}(e_k) & \\
&=&\sum_{i=0}^{\mu-1}\delta(a_i)\delta^{i}(e_k)+\sum_{i=1}^{\mu-1}a_{i-1}\delta^{i}(e_k) +   a_{\mu-1}\delta^{\mu}(e_k) &\\
&=&\sum_{i=0}^{\mu-1}\delta(a_i)\delta^{i}(e_k)+\sum_{i=1}^{\mu-1}a_{i-1}\delta^{i}(e_k) +   a_{\mu-1}\sum_{i=0}^{\mu-1}a_i\delta^{i}(e_k).&\\
\end{array}$$
The third step holds by (\ref{eq mu}). Therefore the $(\mu+2)$th column of $E^r$ is linear combination of the first $\mu$ columns. Repeating the process we obtain that every column from the $(\mu+1)$th to the $r$th one is a linear combination of the first $\mu$ columns, which implies that $rk\,E^r=\mu$. Since $E^r$ has $v$ rows, $\mu\leq v$. Hence, by Lemma \ref{lemma 4}, $rk\,S^r=\mu$.

\end{proof}

We are now ready for computing $\rho$.

\begin{proposition} \label{lemma 6} The reduced column echelon form of $S^\tau$ is of the form
\[
rcef(S^\tau)=\left( \begin{array}{c|c}
I_{\mu}&\\
\hhline{-|~}
a_0\cdots a_{\mu-1}& 0_{(\tau+1)\times (\tau-\mu)}\\
\hhline{-|~}
H' &\\
 \end{array}\right),
 \]
where $I_\mu$ is the $\mu \times \mu$ identity matrix and $a_0,\dots,a_{\mu-1}\in \mathbb{F}_q(z)$ are such that $\rho=x^\mu-\sum_{i=0}^{\mu-1}a_ix^i$.
\end{proposition}

\begin{proof} By Lemma \ref{lemma 5}, $rk\,S^{\tau}=\mu=rk\,S^\mu$ which implies, since $S^{\mu}$ consists of the first $\mu$ columns of $S^\tau$,  that
$$
rcef(S^{\tau})=\left( \begin{array}{c|c}
rcef(S^{\mu})&0_{(\tau+1)\times (\tau-\mu)}
 \end{array}\right).$$

Since $S^\mu$ consists of the first $\tau+1$ rows of $\Sigma E^{\mu}$ and both have the same rank $\mu$, we get that $rcef(S^\mu)$ is composed by the first $\tau+1$ rows of $rcef(\Sigma E^\mu)$. By Proposition \ref{Prop 3}, $\mathfrak{v}(\mathcal{R}\rho)$ is the left kernel of the matrix $rcef(\Sigma E^\mu)$. A non zero vector in the left kernel of the matrix 
\begin{equation}\label{eq system}
\left(\begin{array}{l|r}
\begin{array}{c}rcef(\Sigma E^\mu) \end{array}
&\begin{array}{c}0 \\
 \hline I_{p -(\mu+1)} \end{array}
 \end{array}\right)
 \end{equation}
is a non zero element of $\mathfrak{v}(\mathcal{R}\rho)$ whose $p-(\mu+1)$ last coordinates are zero. Since $\rho$ has degree $\mu$, and its degree is minimal among the nonzero elements of $\mathcal{R}\rho$, we have that $\mathfrak{v}(\rho)$ is the unique element, up to scalar multiplication, of the left kernel of the matrix \eqref{eq system}. 
 
 Let us consider $S_0^\mu$ formed by the first $\mu+1$ rows of $S^\mu$. Therefore 

$$rcef(S_0^\mu)=\left( \begin{array}{c}
I_{\mu}\\
\hline
H' \\
 \end{array}\right).$$
Further column reductions using the identity matrix in the left block of the matrix in \eqref{eq system}, allow us to derive that $\rho$ is also the unique non zero solution, up to scalar multiplication, of the homogeneous system

\begin{equation}\label{eq matrix}
X\left( \begin{array}{c|c}

 rcef( S_0^\mu) & 0 \\
\hline
0 &
 I_{p -(\mu+1)}\\

 \end{array}\right)=0.
 \end{equation}

The size of $rcef( S_0^\mu)$ is $(\mu+1)\times \mu$. Moreover, $rk\,S_0^\mu=\mu$ because the space solutions of (\ref{eq matrix}) has dimension $1$. Then there is only one row of $rcef( S_0^\mu)$ without a pivot. If this row is not the last one, then there exists a non zero polynomial in $\mathcal{R}\rho$ of degree strictly below $\mu$, which is a contradiction. Therefore

$$rcef(S_0^\mu)=\left( \begin{array}{c}
I_{\mu}\\
\hline
a_0\cdots a_{\mu-1} \\
 \end{array}\right).$$
 
 Finally, $(-a_0,\dots, -a_{\mu-1},1,0,\dots,0)$ is a non zero solution of (\ref{eq matrix}). Hence $\rho=x^\mu-\sum_{i=0}^{\mu-1}a_ix^i$.
\end{proof}

Our next task is to compute the locator polynomial $\lambda$ from $\rho$. We need some auxiliary results. 

\begin{lemma} \label{Lemma 6} If $\mathfrak{v}(\mathcal{R}\rho)$ is the left kernel of a matrix $H \in \mathcal{M}_{p \times \mu}(\Fz)$, then $H=\Sigma B$ for some full rank matrix $B\in \mathcal{M}_{v \times \mu}(\mathbb{F}_q(z))$ which has no zero row. 
\end{lemma} 

\begin{proof} 
According to Proposition \ref{Prop 3}, $\mathfrak{v}(\mathcal{R}\lambda) \subseteq \mathfrak{v}(\mathcal{R}\rho) \subseteq \Fz^p$. Moreover, the dimension over $\Fz$ of the first subspace is $p-v$, while that of the second one is $p-\mu$. Basic Linear Algebra  tells us that there exists a full rank $(v \times \mu)$-matrix $B$ such that $H = \Sigma B$. 
By Proposition \ref{Prop 3}, $\mathfrak{v}(\mathcal{R}\rho)$ is also the left kernel of $\Sigma E^\mu$. Then there exists a non singular $(\mu \times \mu)$-matrix $P$ such that $\Sigma E^\mu P =H = \Sigma B$. Since $\Sigma $ is full rank, $E^\mu P = B$. Thus, $B$ is obtained from $E^{\mu}$ by elementary operations on its columns. Since $E^{\mu}$ has no zero row, we get the same property for $B$. 
\end{proof}

At this stage, it is convenient to introduce some notation. For a non empty $T \subseteq \{0, \dots, p-1\}$ let us denote
\[
g_T = [\{x-L(\delta^j(\alpha))\}_{j\in T}]_\ell. 
\]
The polynomials of this kind will be said to be \emph{fully $\alpha$--decomposable}. Since the least  left  common multiple corresponds to the intersection of left ideals, we get that $[g_{T_1},g_{T_2}]_\ell   = g_{T_1 \cup T_2}$ for all $T_1, T_2 \subseteq \{0, \dots, p-1\}$. On the other hand, $g_{T_1 \cap T_2}$ is clearly a common right divisor of $g_{T_1}$ and $g_{T_2}$. By Proposition \ref{DCCdim}, 
 \begin{multline}
 \# T_1 \cap T_2 = \deg (g_{T_1 \cap T_2}) \leq \deg (g_{T_1},g_{T_2})_r  = \\ \deg g_{T_1} + \deg g_{T_2} - \deg [g_{T_1},g_{T_2}]_\ell = \# T_1 + \# T_2 - \# T_1 \cup T_2. 
\end{multline}
Hence, $(g_{T_1}, g_{T_2})_r = g_{T_1 \cap T_2}$ is a fully $\alpha$-decomposable polynomial.

Given a matrix $A$ with coefficients in $\mathbb{F}_q(z)$, by $\ker (\cdot A)$ we denote its left kernel.  

\begin{proposition} \label{prop fully decomposable left} Let $\lambda'\in \mathcal{R}$ be a fully $\alpha-$decomposable polynomial which is a left multiple of $\rho$. Then $\lambda|_r\lambda'$.
\end{proposition}

\begin{proof} 

By the discussion above, we can write $(\lambda,\lambda')_r = g_T$ for some $T \subseteq \{k_1,\dots, k_v\}$.  Set $t = \#T$ and relabel the elements of $\{k_1,\dots,k_v\}$ in such a way that $T = \{k_1, \dots, k_t\}$. It suffices if we prove that $g_T = \lambda$, that is, $t = v$.

We know that $\mathfrak{v}(\mathcal{R}\lambda) = \ker (\cdot \Sigma)$ and $rk\, \Sigma = v$.  Since $\mathfrak{v}(\mathcal{R}\lambda)\subseteq \mathfrak{v}(\mathcal{R}g_T)$, we have that $\mathfrak{v}({\mathcal{R}g_T}) = \ker (\cdot \Sigma Q)$ for some full rank matrix $Q$. Moreover, $\dim_{\Fz} \mathcal{R}g_T = p - \deg g_T = p - t$, by Proposition \ref{DCCdim}. Hence, $rk\, \Sigma Q = t $. 

Analogously, since $\mathcal{R}g_T\subseteq \mathcal{R}\rho$, there exists a full rank matrix $Q'$ such that $\mathfrak{v}(\mathcal{R}\rho)$ is the left kernel of $\Sigma Q Q'$. By Lemma \ref{Lemma 6},  $\Sigma Q Q'=\Sigma B$, where $B$ has full rank and no zero row. Hence, $QQ'=B$, because $\Sigma$ defines a surjective linear map, and thus, $Q$ has no zero row.

Assume $t < v$. We can split the matrices $\Sigma$ and $Q$ as 
\[
\Sigma =\left(\Sigma_0\mid \Sigma_1\right)\text{ and } Q=\left( \begin{array}{c}
Q_0\\
\hline
Q_1 \\
 \end{array}\right),
 \]
 where $\Sigma_0$ encompasses the first $t$ columns of $\Sigma$ and $Q_0$ the first $t$ rows of $Q$. Let $c_1, \dots, c_v$ denote the columns of $\Sigma$. Now, if $j \leq t$, then $c_j$ depends linearly from the columns of $\Sigma Q$, because $\ker (\cdot \Sigma Q) \subseteq \ker (\cdot c_j)$. Therefore,
 \[
 t = rk\, (\Sigma Q \mid \Sigma_0) = rk\, (\Sigma_0 Q_0 + \Sigma_1 Q_1\mid \Sigma_0) = rk\, (\Sigma_1 Q_1 \mid \Sigma_0).
 \]
But $t = rk\, \Sigma_0$, hence every column of $\Sigma_1 Q_1$ depends linearly from the columns of $\Sigma_0$. Let $(b_{t+1}, \dots, b_v)^{tr}$ be a column of $Q_1$ with $b_v \neq 0$. We have that
\[
\Sigma_1 (b_{t+1},\dots,b_v)^{tr} = \sum_{i=t+1}^vc_ib_i
\]
depends linearly from the columns of $\Sigma_0$. This would prove that the columns of $\Sigma$ are linearly dependent, in contradiction with Lemma \ref{lemma 3}. 
\end{proof}

Proposition \ref{prop fully decomposable left} implies that $\rho = \lambda$ in most cases, and this leads to a decoding algorithm (see Algorithm \ref{Algorithm0}).  To this end, set
\begin{equation}\label{N}
N=\left( \begin{array}{cccc}
N_{0}(L(\alpha)) & N_{0}(L(\delta(\alpha))) & \cdots &N_{0}(L(\delta^{p-1}(\alpha)))  \\
N_{1}(L(\alpha))&N_{1}(L(\delta(\alpha)))& \cdots & N_{1}(L(\delta^{p-1}(\alpha)))  \\
\vdots &\vdots & \ddots& \vdots\\

N_{p-1}(L(\alpha))& N_{p-1}(L(\delta(\alpha))) & \cdots &  N_{p-1}(L(\delta^{p-1}(\alpha)))
\end{array} \right),
\end{equation}

and observe that, by virtue of \eqref{root2}, for every $f \in \mathcal{R}$ we have
\begin{equation}\label{fevaluacion}
\mathfrak{v}(f)N = (f[\alpha], f[\delta(\alpha)], \dots, f[\delta^{p-1}(\alpha)]).
\end{equation}

\begin{algorithm}[h]
\KwIn{A received transmission $y=(y_0,\dots,y_{p-1})\in \Fz^p$ with no more than $\tau$ errors.}
\KwOut{The error $e=(e_0,\dots,e_{p-1})$ such that $y-e\in C$.}

\nl \textbf{for } $0\leq i\leq 2\tau-1$ \textbf{do}\\
\nl   $s_i\gets \sum_{j=0}^{p-1}y_jN_j(L(\delta^i(\alpha)))$\\
 \nl \textbf{if  } $s_i=0$ \text{ for all }  $0\leq i\leq 2\tau-1$ \textbf{then }\\
 \nl \textbf{return } 0. \\
 \nl \label{line 5} Get $S^{\tau}$ by recursion from $y$\\
 \nl Compute\[rcef(S^\tau)=\left( \begin{array}{c|c}
I_{\mu}&\\
\hhline{-|~}
a_0\cdots a_{\mu-1}& 0_{(\tau+1)\times (\tau-\mu)}\\
\hhline{-|~}
H' &\\
 \end{array}\right).\] \\
 \nl \label{line 7} $\rho=(\rho_0,\dots,\rho_{\mu})\gets(-a_0,\dots,-a_{\mu-1},1)$ and $\rho_{N}\gets (\rho_0,\dots,\rho_{\mu},0,\dots,0)N$ \\
\nl \label{line 8} $\{k_1,\dots,k_v\}\gets \text{ zero coordinates  of }\rho_{N}$ \\ 
\nl \textbf{if  } $\mu \neq v$ \textbf{then }\\
\textbf{return} \emph{decoding failure}\\
\textbf{end} \\
\nl \label{line 11} Find $(x_1,\dots,x_v)$ such that $(x_1,\dots,x_v)(\Sigma^{v-1})^{tr}=(\alpha s_0,\delta(\alpha) s_1,\dots, \delta^{v-1}(\alpha) s_{v-1})$\\
\nl \textbf{return} $(e_0,\dots,e_{p-1})$ with $e_i=x_i$ for $i\in \{k_1,\dots,k_v\}$, and zero otherwise. 
    \caption{{PGZ decoding algorithm with unlikely decoding failure} 
\label{Algorithm0}}

\end{algorithm}

\begin{theorem}
Assume that the error vector $e$ has $v$ non-zero positions. If $v \leq \tau$, then Algorithm \ref{Algorithm0} correctly finds the error vector $e$ unless its components $e_{k_1}, \dots, e_{k_v}$ are linearly dependent over $K = \mathbb{F}_q(z^p)$. 
\end{theorem}
\begin{proof}
After the initial settings, Line \ref{line 5} computes $S^\tau$, by virtue of  Lemma \ref{calcular} and, hence,  Line \ref{line 7} computes a right divisor  $\rho=\Sigma_{i=0}^\mu\rho_ix^i$ of the  error locator, by Proposition \ref{Prop 3} and Lemma \ref{lemma 6}. Line \ref{line 8} computes correctly, according to \eqref{fevaluacion},  indices $\{k_1, \dots, k_u \}$ such that $\delta^{k_i}(\alpha)$ is a right root of $\rho$.  If $u = \mu$, then $\rho$ is fully $\alpha$--decomposable and, by Proposition \ref{prop fully decomposable left}, $\rho = \lambda$ (and $u = v$). By Proposition \ref{error values}, Line \ref{line 11} computes the error values.  

If $u < \mu$, then $\dim (\mathcal{R}\lambda) < \dim (\mathcal{R}\rho)$, so $v < \mu$. Proposition \ref{Prop 3}  tells us that $E^v$ is a singular matrix. By Lemma \ref{lemma 3}, $\{e_{k_1}, \dots, e_{k_v}\}$ is $K$--linearly dependent. 
\end{proof}

Although the condition that leads to the decoding failure output in Algorithm \ref{Algorithm0} rarely occurs for a random error vector, we will show that it is possible to design an improved version which does always compute the error vector $e$. To this end, consider, for $f \in \mathcal{R}$ of degree $m$,  the matrix
\[
M_f = \left( \begin{array}{c} 
\mathfrak{v}(f) \\
\mathfrak{v}(xf) \\
\vdots \\
\mathfrak{v}(x^{p - 1 - m}f)
\end{array} \right), 
\]
to be used in Algorithm \ref{Algorithm}. Observe that $f, xf, \dots, x^{p-1-m}f$ are polynomials of different degrees $m, \dots, p-1$, so they are $\Fz$--linearly independent in $\mathcal{R}$. Since the dimension of $\mathcal{R}f$ is $p-m$, we get that they are a basis and, hence, the rows of $M_f$ give a basis of $\mathfrak{v}(\mathcal{R}f)$.

\begin{algorithm}[h]
\KwIn{A received transmission $y=(y_0,\dots,y_{p-1})\in \Fz^p$ with no more than $\tau $ errors.}
\KwOut{The error $e=(e_0,\dots,e_{p-1})$ such that $y-e\in C$.}

\nl \textbf{for } $0\leq i\leq 2\tau-1$ \textbf{do}\\
\nl   $s_i\gets \sum_{j=0}^{p-1}y_jN_j(L(\delta^i(\alpha)))$\\
 \nl \textbf{if  } $s_i=0$ \text{ for all }  $0\leq i\leq 2\tau-1$ \textbf{then }\\
 \nl \textbf{return } 0. \\
 \nl Get $S^{\tau}$ by recursion from $y$\\
 \nl Compute\[rcef(S^\tau)=\left( \begin{array}{c|c}
I_{\mu}&\\
\hhline{-|~}
a_0\cdots a_{\mu-1}& 0_{(\tau+1)\times (\tau-\mu)}\\
\hhline{-|~}
H' &\\
 \end{array}\right).\] \\
 \nl \label{line VII} $\rho=(\rho_0,\dots,\rho_{\mu})\gets(-a_0,\dots,-a_{\mu-1},1)$ and $\rho_{N}\gets (\rho_0,\dots,\rho_{\mu},0,\dots,0)N$ \\
\nl \label{line VIII} $\{k_1,\dots,k_v\}\gets \text{ zero coordinates  of }\rho_{N}$ \\ 
\nl \textbf{if  } $\mu \neq v$  \textbf{then }\\
\nl Calculate $M_\rho N$\\
\nl  $N_\rho\gets M_\rho N$ \\
\nl  $H_\rho\gets rref(N_\rho)$ \\
\nl  \label{line 13} $H'$ gets the matrix obtained removing all rows of $H_\rho$ different from $\varepsilon_i$ for any $i$. \\
\nl $\{k_1,\dots,k_v\}\gets \text{ zero columns
 of }H'$ \\
\textbf{end} \\
\nl \label{line 15} Find $(x_1,\dots,x_v)$ such that $(x_1,\dots,x_v)(\Sigma^{v-1})^{tr}=(\alpha s_0,\delta(\alpha) s_1,\dots, \delta^{v-1}(\alpha) s_{v-1})$\\
\nl \textbf{return} $(e_1,\dots,e_{p-1})$ with $e_i=x_i$ for $i\in \{k_1,\dots,k_v\}$, and zero otherwise. 
    \caption{PGZ full decoding algorithm} 
\label{Algorithm}

\end{algorithm}

\begin{theorem} 
Assume that the error vector $e$ has $v$ non-zero positions. If $v \leq \tau$, then Algorithm \ref{Algorithm} correctly finds the error vector $e$.
\end{theorem}

\begin{proof} 
The output \emph{decoding failure} in Algorithm \ref{Algorithm0} appears when $v \neq \mu$. We will show that, if  the condition $v\neq \mu$ holds, then it is possible to compute a new set of error positions $\{k_1, \dots, k_v \}$ which becomes complete, proving thus the correctness of Algorithm \ref{Algorithm}.  

 Since $H_\rho$ is the reduced row echelon form of $N_\rho = M_\rho N$, the rows of $H_\rho N^{-1}$ form a basis of $\mathfrak{v}(\mathcal{R}\rho)$ as an $\Fz$--vector space. Therefore, the rows of $H'N^{-1}$ generate a vector subspace $V'$ of $\mathfrak{v}(\mathcal{R}\rho)$. Let $f \in \mathcal{R}$ be such that $\mathfrak{v}(f)$ is one of these rows. The components of $\mathfrak{v}(f)N$ corresponding to the zero columns of $H'$ are zero, so $f \in \mathcal{R}\lambda'$, where $\lambda' =  [x-L(\delta^{k_1}(\alpha)), \dots, x-L(\delta^{k_v}(\alpha))]_\ell$. Consequently, the row space of $H'N^{-1}$ is contained in $\mathfrak{v}(\mathcal{R}\lambda')$. Both have dimension $p-v$ as $\Fz$--vector spaces, so they are equal. 

By Proposition \ref{prop fully decomposable left}, $\lambda\mid_r\lambda'$. Suppose that $\lambda\neq \lambda'$, then the matrix $H_\lambda=rref(M_\lambda N)$ contains an additional row $\varepsilon_d$ not in $H'$. Since $\mathcal{R}\lambda\subseteq\mathcal{R}\rho$, $rk\left( \begin{array}{c}
H_\rho\\
\hline
\varepsilon_d \\
 \end{array}\right)=rk(H_\rho)$. This implies that $\varepsilon_d$ must be a row of $H_\rho$, so it is not removed in Line \ref{line 13} of the Algorithm \ref{Algorithm}. Hence $\varepsilon_d$ belongs to $H'$, a contradiction. Then $\lambda=\lambda'$ and the error positions are computed. 
\end{proof}

\section{Examples}\label{Ex}

The last section is devoted to giving some examples. In the first one (Example \ref{exfallo}), we show that the outcome  ``decoding failure''  of Algorithm \ref{Algorithm0} is possible for some corrupted messages, so that Algorithm \ref{Algorithm} is required. It also includes examples of messages which are successfully corrected by Algorithm \ref{Algorithm0}. This is also the case of Example \ref{exsinfallo}.
\begin{example}\label{exfallo}
Let us consider the prime field with $p$ elements $\mathbb{F} = \mathbb{F}_{p}$, $\mathbb{F}(z)$ the field of rational functions over $\mathbb{F}$, and the standard derivation $\delta: \mathbb{F}(z) \rightarrow \mathbb{F}(z)$, that assigns to each rational function its derivative. Clearly,  $\delta$ satisfies the polynomial equation $\delta^{p}  = 0$, since $\delta^{p}(z) = 0$. Set  $\mathcal{R} = \mathbb{F}(z)[x;\delta]/\langle x^{p}\rangle$ as the word ambient ring for differential  convolutional codes. Let set $\alpha = 1/z\in \mathbb{F}(z)$, hence $\delta^i(\alpha)= (-1)^i i!/z^{i+1}$ for $i=0,\ldots, p-1$ and then $\{\delta^i(\alpha)\}_{i=0,\ldots, p-1}$ is  basis of $\mathbb{F}(z)$ over the constant subfield $\mathbb{F}(z^{p})$. Consequently, since $L(\delta^i(\alpha)) = -(i+1)/z$ for $i=0,\dots, p-1$,
$$x^{p} = \lclm{x+\frac{i}{z}, x}^{i=1,\ldots ,p-1}.$$
Moreover, the $(i,j)$-th component of the matrix $N$ defined in \eqref{N} is given in this case by  $$N_{ij}=(-1)^i \frac{(i+j)!}{j!} \frac{1}{z^i},$$ so that $N$ presents an upper triangular matrix form. For instance, when $p=11$, $N$ turns out to be

$$
\begin{tikzpicture}
\matrix [matrix of math nodes, left delimiter=(, right delimiter=)] (m)
{
 1 & 1 & 1 & 1 & 1 & 1 & 1 & 1 & 1 & 1 & 1 \\ 
\frac{10}{z} & \frac{9}{z} & \frac{8}{z} & \frac{7}{z} & \frac{6}{z} & \frac{5}{z} & \frac{4}{z} & \frac{3}{z} & \frac{2}{z} & \frac{1}{z} & 0 \\ 
\frac{2}{z^{2}} & \frac{6}{z^{2}} &  \frac{1}{z^{2}} & \frac{9}{z^{2}} & \frac{8}{z^{2}} & \frac{9}{z^{2}} & \frac{1}{z^{2}} & \frac{6}{z^{2}} & \frac{2}{z^{2}} & 0 & 0 \\ 
\frac{5}{z^{3}} & \frac{9}{z^{3}} & \frac{6}{z^{3}} & \frac{1}{z^{3}} & \frac{10}{z^{3}} & \frac{5}{z^{3}} & \frac{2}{z^{3}} & \frac{6}{z^{3}} & 0 & 0 & 0 \\ 
\frac{2}{z^{4}} & \frac{10}{z^{4}} & \frac{8}{z^{4}} & \frac{4}{z^{4}} & \frac{8}{z^{4}} & \frac{10}{z^{4}} & \frac{2}{z^{4}} & 0 & 0 & 0 & 0 \\ 
\frac{1}{z^{5}} & \frac{6}{z^{5}} & \frac{10}{z^{5}} & \frac{1}{z^{5}} & \frac{5}{z^{5}} & \frac{10}{z^{5}} & 0 & 0 & 0 & 0 & 0 \\ 
\frac{5}{z^{6}} & \frac{2}{z^{6}} & \frac{8}{z^{6}} & \frac{2}{z^{6}} & \frac{5}{z^{6}} & 0 & 0 & 0 & 0 & 0 & 0 \\ 
\frac{9}{z^{7}} & \frac{6}{z^{7}} & \frac{5}{z^{7}} & \frac{2}{z^{7}} & 0 & 0 & 0 & 0 & 0 & 0 & 0 \\ 
\frac{5}{z^{8}} & \frac{1}{z^{8}} & \frac{5}{z^{8}} & 0 & 0 & 0 & 0 & 0 & 0 & 0 & 0 \\ 
\frac{10}{z^{9}} & \frac{1}{z^{9}} & 0 & 0 & 0 & 0 & 0 & 0 & 0 & 0 & 0 \\ 
\frac{10}{z^{10}} & 0 & 0 & 0 & 0 & 0 & 0 & 0 & 0 & 0 & 0 \\
};  
\draw (-3.2,-0.25) -- (0.5,-0.25) -- (0.5,3.5) -- (-3.2,3.5) -- (-3.2,-0.25);
\draw[dashed] (-3.2,-0.25) -- (-3.2,-0.85) -- (0.5,-0.85) -- (0.5,-0.25);
\end{tikzpicture} 
$$
Consider the RS differential convolutional code $\mathfrak{v}(\mathcal{R}g)$, where 
\[
g = \lclm{x+\frac{1}{z}, x+\frac{2}{z}, x+\frac{3}{z}, x+\frac{4}{z}, x+\frac{5}{z}, x+\frac{6}{z}}. 
\]
This code has designed distance $d=7$.  We can compute $g$ by means of the non-commutative extended Euclidean algorithm (see, e.g. \cite[Ch. I, Theorem 4.33]{Bueso/alt:2003}). Alternatively, since $g$ is characterized as the unique monic polynomial of degree $6$ belonging to $\mathcal{R}g$, we see that it can be computed by solving the linear system $(g_0, \dots, g_{5})A = - b$, where  
$A$ is the principal $6\times 6$ of $N$ and $b$ is the vector formed by the first six components of the seventh row. We get thus
\[
g= x^{6} + \frac{3}{z} x^{5} + \frac{10}{z^{2}} x^{4} + \frac{2}{z^{3}} x^{3} + \frac{10}{z^{4}} x^{2} + \frac{8}{z^{5}} x + \frac{5}{z^{6}}. 
\]

In this setting, suppose we need to transmit the message 
$m=(1,z,0,0,z^4) \equiv 1+zx+z^4x^4\in \mathcal{R}$, so that we encode it and get 
\[
mg=z^{4} x^{10} + 3 z^{3} x^{9} + 9 z^{2} x^{8} + 3 z x^{7} + 3 x^{6} + \frac{5}{z} x^{5} + \frac{8}{z^{2}} x^{4} + \frac{7}{z^{3}} x^{3} + \frac{3}{z^{4}} x^{2} + \frac{5}{z^{5}} x + \frac{3}{z^{6}},
\]
or, as a list,
$$\left(\frac{3}{z^{6}}, \frac{5}{z^{5}}, \frac{3}{z^{4}}, \frac{7}{z^{3}}, \frac{8}{z^{2}}, \frac{5}{z}, 3, 3 z, 9 z^{2}, 3 z^{3}, z^{4}\right).$$

Suppose now that, after the transmission, it is received the list
$$\left(\frac{3}{z^{6}}, \frac{5}{z^{5}}, \frac{3}{z^{4}}, \frac{7}{z^{3}}, \frac{8}{z^{2}}, \frac{5}{z}, 0, 3 z, 0, 3 z^{3}, z^{4}\right).$$
That is, the received message contains two errors $e_1 = 8$ and $e_2 =2 z^2$ at positions 6 and 8, respectively. Viewed as a polynomial,
\[
y = z^{4} x^{10} + 3 z^{3} x^{9} + 3 z x^{7} + \frac{5}{z} x^{5} + \frac{8}{z^{2}} x^{4} \\ + \frac{7}{z^{3}} x^{3} + \frac{3}{z^{4}} x^{2} + \frac{5}{z^{5}} x + \frac{3}{z^{6}}.
\]
 In this case, the syndrome matrix and its reduced column echelon form are provided by 
$$S^\tau=\left(\begin{array}{ccc}
\frac{6}{z^{7}} & \frac{9}{z^{8}} & \frac{9}{z^{9}} \\ [0.8ex]
\frac{4}{z^{8}} & \frac{7}{z^{9}} & \frac{7}{z^{10}} \\ [0.8ex]
\frac{5}{z^{9}} & \frac{7}{z^{10}} & \frac{7}{z^{11}} \\ [0.8ex]
\frac{3}{z^{10}} & 0 & 0
\end{array}\right) \text{ and }
rcef(S^\tau)=\left(\begin{array}{cc|c}
1 & 0 & 0 \\
0 & 1 & 0 \\ [0.8ex]
\frac{3}{z^{2}} & \frac{5}{z} & 0 \\ [0.8ex]
\frac{9}{z^{3}} & \frac{1}{z^{2}} & 0
\end{array}\right),
$$
respectively. So that the vector $\rho$ becomes
$$\left(\frac{8}{z^{2}}, \frac{6}{z}, 1, 0, 0, 0, 0, 0, 0, 0, 0\right)$$
and then, multiplying by $N$, we find that
$$\rho_N=\left(\frac{4}{z^{2}},\,\frac{2}{z^{2}},\,\frac{2}{z^{2}},\,\frac{4}{z^{2}},\,\frac{8}{z^{2}},\,\frac{3}{z^{2}},\,0,\,\frac{10}{z^{2}},\,0,\,\frac{3}{z^{2}},\,\frac{8}{z^{2}}\right).$$
We may observe two zeroes at positions 6 and 8. This is compatible with the degree of $\rho$ as a polynomial, so the received message surely contains two errors at positions 6 and 8. Finally, the values of the errors are given by solving the linear system
$$\left(\begin{array}{cc}
\frac{5}{z^{7}} & \frac{9}{z^{8}} \\ [0.8ex]
\frac{5}{z^{9}} & \frac{10}{z^{10}}
\end{array}\right)^{tr}\left(\begin{array}{c} e_1 \\[0.8ex] e_2 \end{array}\right) = 
\left(\begin{array}{c} \frac{6}{z^{7}} \\ [0.8ex] \frac{4}{z^{8}}\end{array}\right).$$
Concretely, $e_1=8$ and $e_2=2z^2$, and the error, as a polynomial, is given by $e=8x^6+2z^2x^8$. Then $mg = y - e$, and dividing by the right by $g$, we may recover the message $m$.

Alternatively, under the same conditions, suppose that we receive the list
$$\left(\frac{3}{z^{6}}, \frac{z^{5} + 5}{z^{5}}, \frac{3}{z^{4}}, \frac{7}{z^{3}}, \frac{8}{z^{2}}, \frac{5}{z}, 0, 3 z, 9 z^{2}, 0, z^{4}\right)$$
i.e. the transmission contains three errors $e_1=1$, $e_2= 8$ and $e_3=8z^3$ at positions 1, 6 and 9, respectively. Now, the syndrome matrix, in its column reduced normal form, is given by
$$rcef(S^\tau)= 
\left(\begin{array}{cc|c}
1 & 0 & 0 \\
0 & 1 & 0 \\ [0.8ex]
\frac{2 z^{5} + 5}{z^{7} + 7 z^{2}} & \frac{9 z^{5} + 6}{z^{6} + 7 z} & 0 \\
\frac{3}{z^{3}} & \frac{8}{z^{2}} & 0
\end{array}\right) $$
and the vector $\rho_N$ becomes
\begin{multline*}\left(\frac{9 z^{5} + 4}{z^{7} + 7 z^{2}},\,\frac{5}{z^{7} + 7 z^{2}},\,\frac{4 z^{5} + 9}{z^{7} + 7 z^{2}},\,\frac{10 z^{5} + 5}{z^{7} + 7 z^{2}},\,\frac{7 z^{5} + 4}{z^{7} + 7 z^{2}},\right. \\ \left. \frac{6 z^{5} + 6}{z^{7} + 7 z^{2}},\,\frac{7 z^{3}}{z^{5} + 7},\,\frac{10 z^{5} + 8}{z^{7} + 7 z^{2}},\,\frac{4 z^{5} + 8}{z^{7} + 7 z^{2}},\,0,\,\frac{9 z^{5} + 6}{z^{7} + 7 z^{2}}\right).\end{multline*}
Since the number of zero components of $\rho_N$ and the degree of $\rho$ are
different, there is a decoding failure. Observe that this is a consequence of the fact
that the determinant of the matrix 
$$\left(\begin{array}{ccc}
e_1 & e_2 & e_3 \\
\delta(e_1) & \delta(e_2) & \delta(e_3) \\
\delta^2(e_1) & \delta^2(e_2) & \delta^2(e_3) 
\end{array}\right)=\left(\begin{array}{ccc}
1 & 8 & 8 z^{3} \\
0 & 0 & 2 z^{2} \\
0 & 0 & 4 z
\end{array}\right)$$
is zero. By the iterated product of $x^i$ by $\rho$ for $i\leq 9$, we then compute a matrix $M_\rho$ whose rows form a basis of the ideal
 $\mathcal{R}\rho$. Hence, we calculate the row reduced normal form of $M_\rho N$, 
 $$H_\rho=\left(\begin{array}{ccccccccccc}
1 & 0 & 0 & 0 & 0 & 0 & 0 & 0 & 0 & 0 & 0 \\
0 & 1 & 0 & 0 & 0 & 0 & 8 z^{5} & 0 & 0 & 0 & 0 \\
0 & 0 & 1 & 0 & 0 & 0 & 0 & 0 & 0 & 0 & 0 \\
0 & 0 & 0 & 1 & 0 & 0 & 0 & 0 & 0 & 0 & 0 \\
0 & 0 & 0 & 0 & 1 & 0 & 0 & 0 & 0 & 0 & 0 \\
0 & 0 & 0 & 0 & 0 & 1 & 0 & 0 & 0 & 0 & 0 \\
0 & 0 & 0 & 0 & 0 & 0 & 0 & 1 & 0 & 0 & 0 \\
0 & 0 & 0 & 0 & 0 & 0 & 0 & 0 & 1 & 0 & 0 \\
0 & 0 & 0 & 0 & 0 & 0 & 0 & 0 & 0 & 0 & 1
\end{array}\right).$$
 The second row of $H_\rho$ is not unitary. Removing that row from $H_\rho$, the resultant matrix has zero columns at positions 1, 6 and 9, the error positions. Finally, by solving the linear system 
 $$\left(\begin{array}{ccc}
\frac{10}{z^{2}} & \frac{2}{z^{3}} & \frac{5}{z^{4}} \\ [0.8ex]
\frac{5}{z^{7}} & \frac{9}{z^{8}} & \frac{5}{z^{9}} \\[0.8ex]
\frac{10}{z^{10}} & \frac{10}{z^{11}} & 0
\end{array}\right)^{tr} 
\left(\begin{array}{c} e_1\\ [0.8ex] e_2 \\ [0.8ex] e_3 \end{array}\right )
= \left(\begin{array}{c}
\frac{10 z^{5} + 10}{z^{7}} \\ [0.8ex] 
\frac{2 z^{5} + 9}{z^{8}} \\ [0.8ex] 
\frac{5 z^{5} + 7}{z^{9}}\end{array}\right),$$
 we find the error values $e_1= 1$, $e_2= 8$ and $e_3=8z^3$, determining completely the error polynomial $e=x+8x^6+8z^3x^9$.
 \end{example}

\begin{example}\label{exsinfallo}
Let us consider the prime field with five elements $\mathbb{F} = \mathbb{F}_{5}$, $\mathbb{F}(z)$ the field of rational functions over $\mathbb{F}$, and the derivation $\delta: \mathbb{F}(z) \rightarrow \mathbb{F}(z)$ defined by $\delta(f) = z f'$ for every $f\in \mathbb{F}(z)$, where $f'$ denotes the derivative of $f$. According to \eqref{polmindelta}, $\delta$ satisfies the polynomial equation $\delta^{5} - \gamma \delta = 0$, where $\gamma = \delta^{5}(z)/\delta (z) = 1$. We may then set the word ambient algebra $\mathcal{R} = \mathbb{F}(z)[x;\delta]/\langle x^{5}-x\rangle$.  

Choose $\alpha = 1/(z+1)\in \mathbb{F}(z)$, which is a cyclic vector for $\delta$. Hence $x^{5}-x$ can be decomposed in  $\mathbb{F}(z)[x;\delta]$ as the least common left multiple
\[
\left[x + \frac{z}{z + 1}, x + \frac{z + 4}{z + 1}, x + \frac{z^{2} + z + 1}{z^{2} + 4}, \right. 
\left. x + \frac{z^{3} + 4 z^{2} + z + 4}{z^{3} + 2 z^{2} + 2 z + 1}, x + \frac{z^{3} + 3 z^{2} + 3 z + 1}{z^{3} + 4 z^{2} + z + 4}\right]_{\ell}.
\]
Let us consider the RS differential convolutional code $C = \mathfrak{v}(\mathcal{R}g)$ of designed distance $3$, where
\small
$$ g  =\lclm{x + \frac{z}{z + 1},
 x + \frac{z + 4}{z + 1}} \\
  = x^{2} + \left(\frac{3 z + 4}{z + 1}\right) x + \frac{2 z^{2}}{z^{2} + 2 z + 1}.$$
  \normalsize

Suppose we need to transmit the message 
$m=(1,0,0) \equiv 1\in \mathcal{R}$, so that the encoded message turns out to be 
$g$ or, viewed as a list,
$$\left(\frac{2 z^{2}}{z^{2} + 2 z + 1}, \frac{3 z + 4}{z + 1}, 1,0,0\right).$$
Now suppose that the received message contains an error  $e =z$ at position 4. That is, we receive the polynomial
$$y = z x^{4} + x^{2} + \left(\frac{3 z + 4}{z + 1}\right) x + \frac{2 z^{2}}{z^{2} + 2 z + 1}.$$
 In this case, the syndrome matrix and its  reduced column echelon form are given by
$$S=\left(\begin{array}{c}
\frac{z^{5} + 4 z^{4} + z^{3} + 4 z^{2}}{z^{5} + 1} \\
\frac{4 z^{2}}{z^{2} + 2 z + 1}
\end{array}\right) \text{ and }
rcef(S)=\left(\begin{array}{c}
1 \\
\frac{4 z^{3} + 2 z^{2} + 2 z + 4}{z^{3} + 4 z^{2} + z + 4}
\end{array}\right),
$$
respectively. So that the vector $\rho$ becomes
$$\left(\frac{z^{3} + 3 z^{2} + 3 z + 1}{z^{3} + 4 z^{2} + z + 4}, 1, 0, 0, 0\right)$$
and then, multiplying by the matrix 
\[
N = \left(\begin{array}{ccc}
1 & 1 & 1 \\
\frac{4 z}{z + 1} & \frac{4 z + 1}{z + 1} & \frac{4 z^{2} + 4 z + 4}{z^{2} + 4} \\
\frac{z^{2} + 4 z}{z^{2} + 2 z + 1} & \frac{z^{2} + z + 1}{z^{2} + 2 z + 1} & \frac{z^{2} + 1}{z^{2} + 2 z + 1} \\
\frac{4 z^{3} + 4 z^{2} + 4 z}{z^{3} + 3 z^{2} + 3 z + 1} & \frac{4 z^{3} + z^{2} + 4 z + 1}{z^{3} + 3 z^{2} + 3 z + 1} & \frac{4 z + 4}{z + 4} \\
\frac{z^{4} + 4 z^{3} + z^{2} + 4 z}{z^{4} + 4 z^{3} + z^{2} + 4 z + 1} & 1 & 1
\end{array}\right. 
\left.\begin{array}{cc}
1 & 1 \\
\frac{4 z^{3} + z^{2} + 4 z + 1}{z^{3} + 2 z^{2} + 2 z + 1} & \frac{4 z^{3} + 2 z^{2} + 2 z + 4}{z^{3} + 4 z^{2} + z + 4} \\
\frac{z^{2} + 2 z + 1}{z^{2} + z + 1} & \frac{z^{2} + 2 z + 1}{z^{2} + 1} \\
\frac{4 z^{2} + 1}{z^{2} + z + 1} & \frac{4 z^{3} + 3 z^{2} + 3 z + 4}{z^{3} + 4 z^{2} + z + 4} \\
1 & 1
\end{array}\right),
\]
it results
\[
\left(\frac{1}{z^{4} + 4},\,\frac{z^{3} + 4 z^{2} + z}{z^{4} + 4},\,\frac{3 z^{2} + 2 z}{z^{3} + z^{2} + z + 1}, \right.  \left. \frac{2 z^{5} + 3 z^{4} + 3 z^{3} + 3 z^{2} + 2 z}{z^{6} + z^{5} + z^{4} + 4 z^{2} + 4 z + 4},\,0\right)
\]
which points out that there is an error at position 4. We, finally, solve the linear system 
$$\frac{z^{4} + 4 z^{3} + z^{2} + 4 z}{z^{5} + 1} e = \frac{z^{5} + 4 z^{4} + z^{3} + 4 z^{2}}{z^{5} + 1}$$
and get that the error value is $e=z$, as expected.
\end{example}

\end{document}